\documentclass[11pt]{amsart}

\onecolumn

\usepackage{fullpage}
\usepackage{amsmath, amssymb, enumerate, color, graphicx, bm, url, mathbbol}

\newtheorem{theorem}{Theorem}[section]
\newtheorem{proposition}[theorem]{Proposition}
\newtheorem{corollary}[theorem]{Corollary}
\newtheorem{lemma}[theorem]{Lemma}

\newtheorem{remark}[theorem]{Remark}

\numberwithin{equation}{section}

\DeclareMathOperator*{\conv}{conv}
\DeclareMathOperator*{\E}{\mathbb{E}}

\DeclareMathOperator*{\sign}{sign}

\DeclareMathOperator*{\diam}{diam}
\DeclareMathOperator*{\dist}{dist}

\def \qed {$\blacksquare$}

\def \N {\mathbb{N}}
\def \R {\mathbb{R}}

\def \P {\mathbb{P}}

\def \one {{\bf 1}}

\def \NN {\mathcal{N}}

\def \b {\beta}

\def \e {\varepsilon}

\def \d {\delta}

\def \s {\sigma}

\def \v {\mathrm{v}}

\def \moo| {\langle}
\def \< {\langle }
\def \> {\rangle }
\def \^ {\widehat}

\newcommand{\norm}[1]{\left \|#1\right \|}
\newcommand{\zeronorm}[1]{\norm{#1}_0}
\newcommand{\onenorm}[1]{\norm{#1}_1}
\newcommand{\twonorm}[1]{\norm{#1}_2}

\newcommand{\opnorm}[1]{\norm{#1}}
\newcommand{\fronorm}[1]{\norm{#1}_F}
\newcommand{\nucnorm}[1]{\norm{#1}_*}

\newcommand{\abs}[1]{\left | #1 \right |}

\renewcommand{\Pr}[1]{\P \left\{ #1 \rule{0mm}{3mm}\right\}}
\def \sech{\text{sech}}
\def \tanh{\text{tanh}}

\newcommand{\vect}[1]{\bm{#1}}
\newcommand{\mat}[1]{\bm{#1}}
\def \va {\vect{a}}

\def \vg {\vect{g}}

\def \vu {\vect{u}}
\def \vv {\vect{v}}

\def \vx {\vect{x}}
\def \vy {\vect{y}}
\def \vz {\vect{z}}
\def \veta {\vect{\eta}}

\def \mA {\mat{A}}

\def \mG {\mat{G}}

\def \mP {\mat{P}}

\def \mX {\mat{X}}

\title{Robust 1-bit compressed sensing and sparse logistic regression: a convex programming approach}

\author{Yaniv Plan}
\author{Roman Vershynin}

\address{Department of Mathematics,
  University of Michigan,
  530 Church St.,
  Ann Arbor, MI 48109, U.S.A.}

\email{\{yplan,romanv\}@umich.edu}


\thanks{Y.P. is supported by an NSF Postdoctoral Research Fellowship under award No. 1103909. 
  R.V. is supported by NSF grants DMS 0918623 and 1001829.}

\date{Submitted February 2012; revised July 2012}


\begin{document}

\begin{abstract}
This paper develops theoretical results regarding noisy 1-bit compressed sensing and sparse binomial regression.
We demonstrate that a single convex program gives an accurate estimate of the signal, 
or coefficient vector, for both of these models.  
We show that an $s$-sparse signal in $\R^n$ can be accurately estimated
from $m = O(s\log (n/s))$ single-bit measurements using a simple convex program. 
This remains true even if each measurement bit is flipped with probability nearly 1/2. 
Worst-case (adversarial) noise can also be accounted for, and uniform results that hold for all sparse inputs are derived as well.  
In the terminology of sparse logistic regression, we show that $O(s \log(2n/s))$ Bernoulli trials are sufficient 
to estimate a coefficient vector in $\R^n$ which is approximately $s$-sparse.  
Moreover, the same convex program works for virtually all generalized linear models, in which the link function may be unknown.
To our knowledge, these are the first results that tie together the theory of sparse logistic regression to 1-bit compressed sensing.  
Our results apply to general signal structures aside from sparsity; one only needs to know the size of the set $K$ where signals reside. 
The size is given by the mean width of $K$, a computable quantity whose square serves as a robust extension of the dimension. 
\end{abstract}

\maketitle

\section{Introduction}

\subsection{One-bit compressed sensing}

In modern data analysis, a pervasive challenge is to recover extremely high-dimensional signals from seemingly inadequate amounts of data.  
Research in this direction is being conducted in several areas including compressed sensing, 
sparse approximation 
and  low-rank matrix recovery. 
The key is to take into account the signal structure, which in essence reduces the dimension of the signal space.  
In compressed sensing and sparse approximation, this structure is {\em sparsity}---we say that a vector in $\R^n$ 
is $s$-sparse if it has $s$ nonzero entries.  In low-rank matrix recovery, one restricts to matrices with low-rank.

The standard assumption in these fields is that one has access to \textit{linear measurements} of the form 
\begin{equation}				\label{eq:linear}
y_i = \< \va_i, \vx \> , \qquad i = 1, 2, \hdots, m
\end{equation}
where $\va_1, \va_2, \hdots, \va_m \in \R^n$ are known measurement vectors and $\vx \in \R^{n}$ is the signal to be recovered.
Typical compressed sensing results state that when $\va_i$ are iid random vectors drawn from a certain  
distribution (e.g.~Gaussian), $m \sim s \log(2n/s)$ measurements suffice
for robust recovery of $s$-sparse signals $\vx$, see \cite{CSBook}.

In the recently introduced problem of {\em 1-bit compressed sensing} \cite{Boufounos2008}, 
the measurements are no longer linear but rather consist of single bits. 
If there is no noise, the measurements are modeled as
\begin{equation} 
\label{eq:noiseless}
y_i = \sign(\< \va_i, \vx\> ), \qquad i =1 , 2, \hdots, m
\end{equation}
where $\sign(x) = 1$ if $x \geq 0$ and $\sign(x) = -1$ if $x < 0$.\footnote{For concreteness, we set $\sign(0) = 1$; this choice is arbitrary and could be replaced with $\sign(0) = -1$.} 
On top of this, noise may be introduced as random or adversarial bit flips.  

The 1-bit measurements are meant to model quantization in the extreme case. 
It is interesting to note that when the signal to noise ratio is low, numerical experiments demonstrate 
that such extreme quantization can be optimal \cite{laska2011regime} when constrained to a fixed bit budget. 
The webpage \verb=http://dsp.rice.edu/1bitCS/= is dedicated to the rapidly growing literature on 1-bit compressed sensing. 
Further discussion of this recent literature will be given in Section \ref{sec:1-bit CS}; 
we note for now that this paper presents the first theoretical accuracy guarantees in the noisy problem using
a polynomial-time solver (given by a convex program).

\subsection{Noisy one-bit measurements}
We propose the following general model for noisy 1-bit compressed sensing. 
We assume that the measurements, or response variables, 
$y_i \in \{-1,1\}$ are drawn independently at random satisfying
\begin{equation}				\label{eq:glm}
\E y_i = \theta(\< \va_i, \vx \> ), \qquad i =1 , 2, \hdots, m
\end{equation}
where $\theta$ is some function, which automatically must satisfy $-1 \leq \theta(z) \leq 1$.
A key point in our results is that {\em $\theta$ may be unknown} or unspecified; 
one only needs to know the measurements $y_i$ and the measurement vectors $\va_i$ in order to recover $\vx$.  Thus, there is an unknown non-linearity in the measurements.  See \cite{boufounos2010reconstruction, laska2011regimethesis} for earlier connections between the 1-bit problem and non-linear measurements.

In compressed sensing it is typical to choose the measurement vectors $\va_i$ at random, see \cite{CSBook}. 
In this paper, we choose $\va_i$ to be independent standard Gaussian random vectors in $\R^n$. 
Although this assumption can be relaxed to allow for correlated coordinates (see Section \ref{sec:covariance}),
discrete distributions are not permitted. Indeed, unlike traditional compressed sensing, 
accurate noiseless 1-bit compressed sensing is provably impossible for some discrete 
distributions of $\va_i$ (e.g.~for Bernoulli distribution, see \cite{pv-1-bit}). 
Summarizing, the model \eqref{eq:glm} has two sources of randomness: 
\begin{enumerate}
  \item the measurement vectors $\va_i$ are independent standard Gaussian random vectors;
  \item given $\{\va_i\}$, the measurements $y_i$ are independent $\{-1,1\}$ valued random variables.
\end{enumerate}

\medskip

Note that \eqref{eq:glm} is the {\em generalized linear model} in statistics, and $\theta$ is known as the inverse of the \textit{link function};
the particular choice $\theta(z) = \tanh(z/2)$ coresponds to logistic regression. 
The statisticians may prefer to switch $\vx$ with $\vect{\beta}$, $\va_i$ with $\vx_i$, $n$ with $p$ and $m$ with $n$, 
but we prefer to keep our notation which is standard in compressed sensing.

\medskip

Notice that in the noiseless 1-bit compressed sensing model \eqref{eq:noiseless}, 
all information about the magnitude of $\vx$ is lost in the measurements.  
Similarly, in the noisy model \eqref{eq:glm} the magnitude of $\vx$ may be absorbed into the definition of $\theta$.  
Thus, our goal will be to estimate the projection of $\vx$ onto the Euclidean sphere, $\vx/\twonorm{\vx}$.  
Without loss of generality, we thus assume that $\twonorm{\vx} = 1$ in most of our discussion that follows.

\medskip

We shall make a single assumption on the function $\theta$ defining the model \eqref{eq:glm}, namely that
\begin{equation}				\label{eq:lambda}
\E \theta(g) g =: \lambda > 0
\end{equation}
where $g$ is standard normal random variable.
To see why this assumption is natural,
notice that $\< \va_i, \vx \> \sim \NN(0,1)$ since 
$\va_i$ are standard Gaussian random vectors and $\twonorm{\vx} = 1$; thus
$$
\E y_i \< \va_i, \vx \> = \E \theta(g) g = \lambda.
$$
Thus our assumption is simply that the 1-bit measurements $y_i$ are positively correlated with the corresponding
linear measurements $\< \va_i, \vx\> $.\footnote{If $\E \theta(g) g < 0$, we could replace $y_i$ with $-y_i$ to change the sign; 
  thus our assumption is really that the correlation is non-zero: $\E \theta(g) g \neq 0$.}
Standard 1-bit compressed sensing \eqref{eq:noiseless} 
is a partial case of model \eqref{eq:glm} with $\theta(z) = \sign(z)$.  In this case $\lambda$ achieves its maximal value: $\lambda = \E |g| = \sqrt{2/\pi}$.  In general, $\lambda$ plays a role similar to a signal to noise ratio.

\subsection{The signal set}						\label{sec:main results}

\medskip

To describe the structure of possible signals, we assume that $\vx$ lies in some set $K \subset B_2^n$ 
where $B_2^n$ denotes the Euclidean unit ball in $\R^n$.
A key characteristic of the size of the signal set is its \textit{mean width} $w(K)$, defined as  \begin{equation}
\label{eq:mean width}
w(K) := \E \sup_{\vx \in K-K} \< \vg, \vx \>
\end{equation}
where $\vg$ is a standard normal Gaussian vector in $\R^n$ and $K-K$ denotes the Minkowski difference.\footnote{Specifically, 
$K-K = \{x-y :\; x,y \in K\}$.}  
  The notion of mean width is closely related to that of the Gaussian complexity, which is widely used in statistical learning theory to measure the size of classes of functions, see \cite{bartlett2003rademacher, LT}.
  An intuitive explanation of the mean width, its basic properties and simple examples are given in Section \ref{sec:mean width}.
The important point is that $w(K)^2$ can serve as the {\em effective dimension} of $K$.

The main example of interest is where $K$ encodes sparsity. 
If $K=K_{n,s}$ is the convex hull of the unit $s$-sparse vectors in $\R^n$, 
the mean width of this set computed in \eqref{eq:sparse mean width} and \eqref{eq:K sparse mean width} is 
\begin{equation}							\label{eq:K mean width sparse intro}
w(K_{n,s}) \sim (s \log(2n/s))^{1/2}.
\end{equation}


\subsection{Main results}

We propose the following solver to estimate the signal $\vx$ from the 1-bit measurements $y_i$. 
It is given by the optimization problem 
\begin{equation}				\label{eq:optimization}
  \max \sum_{i=1}^m y_i \< \va_i, \vx' \> \quad \text{subject to} \quad \vx' \in K.
\end{equation}
This can be described even more compactly as 
\begin{equation}				\label{eq:optimization compact}
\max \< \vy, \mA \vx' \> \quad \text{subject to} \quad \vx' \in K
\end{equation}
where $\mA$ is the $m \times n$ measurement matrix with rows $\va_i$ and $\vy = (y_1,\ldots,y_m)$
is the vector of 1-bit measurements. 

If the set $K$ is convex, \eqref{eq:optimization} is a {\em convex program}, and therefore it can be 
solved in an algorithmically efficient manner. This is the situation we will mostly care about, although our
results below apply for general, non-convex signal sets $K$ as well. 

\begin{theorem}[Fixed signal estimation, random noise]					\label{thm:fixed X}
  Let $\va_1,\ldots,\va_m$ be independent standard Gaussian random vectors in $\R^n$,
  and let $K$ be a subset of the unit Euclidean ball in $\R^n$. 
  Fix $\vx \in K$ satisfying $\twonorm{\vx} = 1$.  
  Assume that the measurements $y_1,\ldots,y_n$ follow the model above.\footnote{Specifically, our assumptions are that 
    $y_i$ are $\{-1,1\}$ valued random variables that are independent given $\{\va_i\}$, and that \eqref{eq:glm} holds 
    with some function $\theta$ satisfying \eqref{eq:lambda}.}  
  Then for each $\b > 0$, with probability at least $1 - 4 \exp(-2 \beta^2)$ the solution $\hat{\vx}$ 
  to the optimization problem \eqref{eq:optimization} 
  satisfies
  $$
  \twonorm{\hat{\vx} - \vx}^2 \leq \frac{8}{\lambda \sqrt{m}}(w(K) + \beta).
  $$
\end{theorem}

As an immediate consequence, we see that 
{\em the signal $\vx \in K$ can be effectively estimated from $m = O(w(K)^2)$ one-bit noisy measurements}. 
The following result makes this statement precise.

\begin{corollary}[Number of measurements]					\label{cor:fixed X}
  Let $\d > 0$ and suppose that
  $$
  m \geq C \d^{-2} w(K)^2.
  $$
  Then, under the assumptions of Theorem \ref{thm:fixed X}, with probability at least $1 - 8 \exp(-c \delta^2 m)$
  the solution $\hat{\vx}$ to the optimization problem \eqref{eq:optimization} satisfies
  $$
  \twonorm{\hat{\vx} - \vx}^2 \leq \delta/\lambda.
  $$
\end{corollary}

Here and in the rest of the paper, $C$ and $c$ denote positive absolute constants whose values may change from instance to instance.

\medskip

Theorem~\ref{thm:fixed X} is concerned with an arbitrary but fixed signal $\vx \in K$, and with a stochastic model on the noise in the measurements. 
We will show how to strengthen these results to cover all signals $\vx \in K$ uniformly, and to allow for a worst-case (adversarial) noise. 
Such noise can be modeled as flipping some fixed percentage of arbitrarily chosen bits, 
and it can be measured using the
Hamming distance $d_H(\tilde{\vy}, \vy) = \sum_{i=1}^m \one_{\{\tilde{y_i} \ne y_i\}}$ between $\tilde{\vy}, \vy \in \{-1,1\}^m$.

We present the following theorem in the same style as Corollary \ref{cor:fixed X}.
\begin{theorem}[Uniform estimation, adversarial noise]					\label{thm:uniform}
  Let $\va_1,\ldots,\va_m$ be independent standard Gaussian random vectors in $\R^n$,
  and let $K$ be a subset of the unit Euclidean ball in $\R^n$. 
  Let $\d > 0$ and suppose that
  \begin{equation}				\label{eq:uniform requirement}
    m \geq C \delta^{-6} w(K)^2.
  \end{equation}
  Then with probability at least $1 - 8 \exp(-c \delta^2 m)$, the following event occurs. 
  Consider a signal $\vx \in K$ satisfying $\|x\|_2 =1$
  and its (unknown) uncorrupted $1$-bit measurements $\tilde{\vy} = (\tilde{y}_1,\ldots,\tilde{y}_m)$ given as 
  $$
  \tilde{y}_i = \sign(\< \va_i, \vx \> ), \quad i = 1, 2, \hdots, m.
  $$
  Let $\vy = (y_1,\ldots,y_m) \in \{-1,1\}^m$ be any (corrupted) measurements satisfying $d_H(\tilde{\vy}, \vy) \leq \tau m$.  
  Then the solution $\hat{\vx}$ to the optimization problem \eqref{eq:optimization} with input $\vy$ satisfies
  \begin{equation} \label{eq:uniform error bound}
  \twonorm{\hat{\vx} - \vx}^2  \leq \delta \sqrt{\log(e/\delta)} + 11\tau \sqrt{\log(e/\tau)}.
  \end{equation}
\end{theorem}

\medskip

This uniform result will follow from a deeper analysis than the fixed-signal result, Theorem~\ref{thm:fixed X}.
Its proof will be based on the recent results from \cite{pv-embeddings} on random hyperplane tessellations of $K$.

\begin{remark}[Sparse estimation]
  A remarkable example is for $s$-sparse signals in $\R^n$. 
  Recalling the mean width estimate \eqref{eq:K mean width sparse intro}, we see that our results 
  above imply that 
  {\em an $s$-sparse signal in $\R^n$ can be effectively estimated from $m = O(s\log(2n/s))$ one-bit noisy measurements}. 
  We will make this statement precise in Corollary~\ref{cor:sparse} and the remark after it.
\end{remark}

\begin{remark}[Hamming cube encoding and decoding]
  Let us put Theorem~\ref{thm:uniform} in the context of coding in information theory.
  In the earlier paper \cite{pv-embeddings} we proved that $K \cap S^{n-1}$ 
  can be almost isometrically embedded into the Hamming cube $\{-1,1\}^m$, 
  with the same $m$ and same probability bound as in Theorem~\ref{thm:uniform}.
  Specifically, one has
  \begin{equation}				\label{eq:embedding}
  \abs{\frac{1}{\pi} d_G(\vx, \vx') - \frac{1}{m} d_H \big(\sign(\mA \vx), \sign(\mA \vx')\big)} \leq \delta
  \end{equation}
  for all   $\vx, \vx' \in K\cap S^{n-1}$.  Above, $d_G$ and $d_H$ denote the geodesic distance in $S^{n-1}$ and the Hamming distance in $\{-1,1\}^m$ respectively,
  see Theorem~\ref{thm:tessellations} below. 
  Thus the embedding $K\cap S^{n-1} \to \{-1,1\}^m$ is given by the map\footnote{The sign function is applied to each coordinate of $\mA \vx$.} 
  $$
  \vx \mapsto \sign(\mA \vx).
  $$
  This map {\em encodes} a given signal $\vx \in K$ into a binary string $\vy = \sign(\mA \vx)$.
  Conversely, one can accurately and robustly {\em decode} $\vx$ from $\vy$ by solving 
  the optimization problem \eqref{eq:optimization}. This is the content of Theorem~\ref{thm:uniform}.
\end{remark}

\begin{remark}[Optimality]
  While the dependence of $m$ on the mean width $w(K)$ in the results above seems to be optimal 
  (see \cite{pv-embeddings} for a discussion), the dependence on the accuracy $\d$ in Theorem \ref{thm:uniform}
  is most likely not optimal. 
  We are not trying to optimize dependence on $\delta$ in this paper, but are leaving this as an open problem.  Nevertheless, in some cases, the dependence on $\delta$ in Theorem \ref{thm:fixed X} is optimal; see Section \ref{sec:1-bit CS} below.
\end{remark}

Theorem~\ref{thm:uniform} can be extended to allow for a random noise together with adversarial noise;
this is discussed in Remark~\ref{rem:random noise} below.

\subsection{Organization}
An intuitive discusison of the mean width along with the estimate \eqref{eq:K mean width sparse intro} 
of the mean width when $K$ encodes sparse vectors is given in Section~\ref{sec:mean width}.  
In Section~\ref{sec:applications} we specialize our results to a variety of (approximately) sparse signal models---1-bit compressed sensing, sparse logistic regression and low-rank matrix recovery.  
In Subsection \ref{sec:covariance}, we extend our results to allow for correlations in the entries of the measurement vectors.  

The proofs of our main results, Theorems \ref{thm:fixed X} and \ref{thm:uniform}, 
are given in Sections~\ref{sec:deducing}---\ref{sec:uniform concentration}.
In Section~\ref{sec:deducing} we quickly reduce these results to the two concentration inequalities that 
hold uniformly over the set $K$---Propositions~\ref{prop:concentration} and \ref{prop:uniform concentration} respectively.
Proposition~\ref{prop:concentration} is proved in Section~\ref{sec:concentration} using standard techniques
of probability in Banach spaces. The proof of Proposition~\ref{prop:uniform concentration} is deeper; it is based
on the recent work of the authors \cite{pv-embeddings} on random hyperplane tessellations. 
The argument is given in Section~\ref{sec:uniform concentration}.

\subsection{Notation} \label{sec:notation}
We write $a \sim b$ if $ca \le b \le Ca$ for some positive absolute constants $c,C$ ($a$ and $b$ may have dimensional dependence).
In order to increase clarity, vectors are written in lower case bold italics (e.g., $\vg$), and matrices are upper case bold italics (e.g, $\mA$). We let $\vg$ denote a standard Gaussian random vector whose length will be clear from context; 
$g$ denotes a standard normal random variable.  $C, c$ will denote positive absolute constants whose values may change from instance to instance. 
Given a vector $\vv$ in $\R^n$ and a subset $T \subset \{1,\ldots,n\}$, we denote by $\vv_T \in \R^T$ the restriction of $\vv$ onto the coordinates in $T$.

$B_2^n$ and $S^{n-1}$ denote the unit Euclidean ball and sphere in $\R^n$ respectively, 
and $B_1^n$ denotes the unit ball with respect to $\ell_1$ norm.
The Euclidean and $\ell_1$ norms of a vector $\vv$ are denoted $\twonorm{\vv}$ and $\onenorm{\vv}$ respectively. 
The number of non-zero entries of $\v$ is denoted $\zeronorm{\vv}$.
The operator norm (the largest singular value) of a matrix $\mA$ is denoted $\opnorm{\mA}$.

\section{Mean width and sparsity} \label{sec:mean width}

\subsection{Mean width}

In this section we explain the geometric meaning of the mean width of a set $K \subset \R^n$ which was defined by 
the formula \eqref{eq:mean width}, and discuss its basic properties and examples.

The notion of mean width plays a significant role in asymptotic convex geometry (see e.g.~\cite{gm-geometry}).
The width of $K$ in the direction of $\veta \in S^{n-1}$ is the smallest width of the slab between two parallel hyperplanes with normals $\veta$ 
that contains $K$. Analytically, the width can be expressed as 
$$
\sup_{\vu \in K} \< \veta, \vu\> - \inf_{\vv \in K} \< \veta, \vv\> 
= \sup_{\vx \in K-K} \< \veta, \vx\> ,
$$ 
see Figure~\ref{fig:mean width}.
Averaging over $\veta$ uniformly distributed in $S^{n-1}$, we obtain the {\em spherical mean width}:
$$
\tilde{w}(K) := \E \sup_{\vx \in K-K} \< \veta, \vx\> .
$$

\begin{figure}[htp]		
  \centering \includegraphics[height=2.6cm]{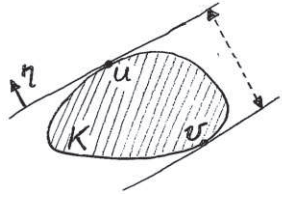} 
  \caption{Width of a set $K$ in the direction of $\veta$ is illustrated by the dashed line.}
  \label{fig:mean width}
\end{figure}

Instead of averaging using $\veta \in S^{n-1}$, it is often more convenient to use a standard Gaussian 
random vector $\vg \in \R^n$. This gives the definition \eqref{eq:mean width} of the {\em Gaussian mean width} of $K$: 
$$
w(K) := \E \sup_{\vx \in K-K} \< \vg, \vx \> .
$$
In this paper we shall use the Gaussian mean width, which we call the ``mean width'' for brevity.
Note that the spherical and Gaussian versions of mean width are proportional to each other. 
Indeed, by rotation invariance we can realize $\veta$ as $\veta = \vg/\|\vg\|_2$ and note that $\veta$ is independent 
of the magnitude factor $\|\vg\|_2$. It follows that $w(K) = \E \|\vg\|_2 \cdot \tilde{w}(K)$. Further, once can use 
that $\E \|g\|_2 \sim \sqrt{n}$ and obtain the useful comparison of Gaussian and spherical versions of mean width:
$$
w(K) \sim \sqrt{n} \cdot \tilde{w}(K).
$$

Let us record some further simple but useful properties of the mean width. 

\begin{proposition}[Mean width]				\label{prop:mean width}
  The mean width of a subset $K \subset \R^n$ has the following properties.
  \begin{enumerate}[1.]
    \item \label{part:affine} The mean width is invariant under orthogonal transformations and translations. 
    \item \label{part:convex hull} The mean width is invariant under taking the convex hull, i.e. $w(K) = w(\conv(K))$. 
    \item \label{part:abs} We have 
      $$
      w(K) = \E \sup_{\vx \in K-K} |\< \vg, \vx \> |.
      $$
    \item \label{part:diam}Denoting the diameter of $K$ in the Euclidean metric by $\diam(K)$, we have 
      $$
      \sqrt{\frac{2}{\pi}} \diam(K) \le w(K) \le n^{1/2} \diam(K).
      $$
    \item \label{part:mean width upper}We have 
      $$
      w(K) \le 2 \E \sup_{\vx \in K} \< \vg, \vx \> 
      \le 2 \E \sup_{\vx \in K} |\< \vg, \vx \> |.
      $$
      For an origin-symmetric set $K$, both these inequalities become equalities.
    \item \label{part:mean width lower} The inequalities in part \ref{part:mean width upper} can be essentially reversed for arbitrary $K$:
     $$
     w(K) \ge \E \sup_{\vx \in K} |\< \vg, \vx \> | - \sqrt{\frac{2}{\pi}} \dist(0,K).
     $$
     Here $\dist(0,K) = \inf_{x \in K} \|x\|_2$ is the Euclidean distance from the origin to $K$.
     In particular, if $0 \in K$ then one has $w(K) \ge \E \sup_{\vx \in K} |\< \vg, \vx \> |$.
  \end{enumerate}
\end{proposition}

\begin{proof}
Parts \ref{part:affine}, \ref{part:convex hull} and \ref{part:mean width upper} are obvious by definition; part \ref{part:abs}
follows by the symmetry of $K-K$. 

To prove part \ref{part:diam}, note that for every $\vx_0 \in K-K$ one has
\begin{equation}							\label{eq:abs moment}
w(K) \ge \E |\< \vg, \vx_0\> | = \sqrt{\frac{2}{\pi}} \, \|\vx_0\|_2.
\end{equation}
The equality here follows because $\< \vg, \vx_0\> $ is a normal random variable with variance $\|\vx_0\|_2$.
This yields the lower bound in part \ref{part:diam}. For the upper bound, we can use part \ref{part:abs}
along with the bound $|\< \vg, \vx \> | \le \|\vg\|_2 \|\vx\|_2 \le \|\vg\|_2 \cdot \diam(K)$ for all $\vx \in K-K$. 
This gives
\begin{align*}
w(K) \le \E \|\vg\|_2 \cdot \diam(K) 
&\le (\E \|\vg\|_2^2)^{1/2} \diam(K) \\
&= n^{1/2} \diam(K).
\end{align*}

To prove part \ref{part:mean width lower}, let us start with the special case where $0 \in K$. 
Then $K-K \supset K \cup (-K)$, thus
$w(K) \ge \E \sup_{\vx \in K \cup (-K)} \< \vg, \vx \>  = \E \sup_{\vx \in K} |\< \vg, \vx \> |$ as claimed.
Next, consider a general case. Fix $\vx_0 \in K$ and apply the previous reasoning for the set $K-\vx_0 \ni 0$.
Using parts \ref{part:affine} and \ref{part:abs} we obtain
\begin{align*}
w(K) = w(K-\vx_0) &\geq \E \sup_{\vx \in K} \abs{\< \vg, \vx - \vx_0\> }\\
&\ge \E \sup_{\vx \in K} \abs{\< \vg, \vx \> } - \E \abs{\< \vg, \vx_0\> }.
\end{align*}
Finally, as in \eqref{eq:abs moment} we note that 
$\E |\< \vg, \vx_0\> | = \sqrt{\frac{2}{\pi}} \, \|\vx_0\|_2$.
Minimizing over $\vx_0 \in K$ completes the proof.
\end{proof}

\begin{example} For illustration, let us evaluate the mean width of some sets $K \subseteq B_2^n$.
  \begin{enumerate}[1.]
    \item If $K = B_2^n$ or $K = S^{n-1}$ then $w(K) = \E \|\vg\|_2 \le (\E \|\vg\|_2^2)^{1/2} = \sqrt{n}$ (and in fact $w(K) \sim \sqrt{n}$). 
    \item If the linear algebraic dimension $\dim(K) = k$ then $w(K) \le \sqrt{k}$.
    \item If $K$ is a finite set, then $w(K) \le C\sqrt{\log|K|}$. 
  \end{enumerate}
\end{example}

\smallskip

\begin{remark}[Effective dimension]
  The square of the mean width, $w(K)^2$, may be interpreted as the {\em effective dimension} 
  of a set $K \subseteq B_2^n$.
  It is always bounded by the linear algebraic dimension (see the example above), but it has the advantage of robustness---a small perturbation of $K$ leads to a small change in $w(K)^2$. 
\end{remark}

In this light, the invariance of the mean width under taking the convex hull (Proposition~\ref{prop:mean width}, part \ref{part:convex hull}) 
is especially useful in compressed sensing, where a usual tactic is to relax the non-convex program to a convex program. 
It is important that in the course of this relaxation, the ``effective dimension'' of the signal set $K$ remains the same.

\smallskip

Mean width of a given set $K$ can be computed using several tools from probability in Banach spaces.  
These include Dudley's inequality,  Sudakov minoration, the Gaussian concentration inequality, Slepian's inequality and the sharp technique of majorizing measures and generic chaining \cite{LT, generic-chaining}.

\subsection{Sparse signal set}
The quintessential signal structure considered in this paper is sparsity.  Thus for given $n \in \N$ and $0 < s \le n$, 
we consider the set 
\[S_{n,s} = \{\vx \in \R^n:\; \zeronorm{\vx} \leq s, \; \twonorm{\vx} \leq 1\}.\]
In words, $S_{n,s}$ consists of $s$-sparse (or sparser) vectors with length $n$ whose Euclidean norm is bounded by 1.

Although the linear algebraic dimension of $S_{n,s}$ is $n$ (as this set spans $\R^n$),
the dimension of $S_{n,s} \cap \{\vx \in \R^n: \zeronorm{\vx} = s\}$ as a manifold with boundary 
embedded in $\R^n$ is $s$.\footnote{Thus $S_{n,s}$ is the union of $s + 1$ manifolds with boundary 
  each of whose dimension is bounded by $s$.}
It turns out that the ``effective dimension'' of $S_{n,s}$ given by the square of its mean width 
is much closer to the manifold dimension $s$ than to the linear algebraic dimension $n$:

\begin{lemma}[Mean width of the sparse signal set]		\label{lem:sparse mean width}
  We have
  \begin{equation}							\label{eq:sparse mean width}
  c s \log(2n/s) \leq w^2(S_{n,s}) \leq C s \log(2n/s).
  \end{equation}
\end{lemma}

\begin{proof}
Let us prove the upper bound. Without loss of generality we can assume that $s \in \N$.  
By representing $S_{n,s}$ as the union of $n \choose s$ $s$-dimensional unit Euclidean balls we see that
$$
w(S_{n,s}) =  \E \max_{|T| = s} \twonorm{\vg_T}.
$$
For each $T$, the Gaussian concentration inequality (see Theorem~\ref{thm:Gaussian concentration} below) yields
$$
\Pr{\twonorm{\vg_T} \geq \E \twonorm{\vg_T} + t} \leq \exp(-t^2/2), \quad t > 0.
$$
Next, $\E \twonorm{\vg_T} \leq (\E \twonorm{\vg_T}^2)^{1/2} = \sqrt{s}$.  Thus the union bound gives 
$$
\Pr{\max_{|T| = s} \twonorm{\vg_T} \geq \sqrt{s} + t} \leq {n \choose s} \exp(-t^2/2)
$$
for $t > 0$.
Note that ${n \choose s} \leq \exp(s \log(en/s))$; integrating gives the desired upper bound in \eqref{eq:sparse mean width}.

The lower bound in \eqref{eq:sparse mean width} follows from Sudakov minoration (see Theorem \ref{thm:sudakov}) 
combined with finding a tight lower bound on the covering number of $S_{n,s}$.  
Since the lower bound will not be used in this paper, we leave the details to the interested reader.
\end{proof}

\section{Applications to sparse signal models} \label{sec:applications}

Our main results stated in the introduction are valid for general signal sets $K$. Now we specialize to the cases where 
$K$ encodes {\em sparsity}. It would be ideal if we could take $K = S_{n,s}$, but this set would not be convex 
and thus the solver \eqref{eq:optimization} would not be known to run in polynomial time.  
We instead take a {\em convex relaxation} of $S_{n,s}$, an effective tactic from the sparsity literature.  
Notice that if $\vx \in S_{n,s}$ then $\onenorm{\vx} \leq \sqrt{s}$ by Cauchy-Schwarz inequality.
This motivates us to consider the convex set
$$
K_{n,s} = \{\vx \in \R^n: \twonorm{\vx} \leq 1, \onenorm{\vx} \leq \sqrt{s}\} = B_2^n \cap \sqrt{s} B_1^n.
$$
$K_{n,s}$ is almost exactly the convex hull of $S_{n,s}$, as is shown in \cite{pv-1-bit}:
\begin{equation}							\label{eq:convexification}
\conv(S_{n,s}) \subset K_{n,s} \subset 2 \conv(S_{n,s}).
\end{equation}
$K_{n,s}$ can be though of a set of {\em approximately sparse} or {\em compressible} vectors. 

If the signal is known to be exactly or approximately sparse, i.e. $\vx \in K_{n,s}$,
we may estimate $\vx$ by solving the convex program 
\begin{equation}					\label{eq:optimization sparse}
 \quad \begin{array}{l}\max \sum_{i=1}^m y_i \< \va_i, \vx' \> \\ \text{subject to} \quad \onenorm{\vx'} \leq \sqrt{s} \quad \text{and} \quad  \twonorm{\vx'} \leq 1.\end{array}
\end{equation}
This is just a restatement of the program \eqref{eq:optimization} for the set $K_{n,s}$.  In our convex relaxation, we do not require that $\hat{\vx} \in S^{n-1}$; this stands in contrast to many previous programs considered in the literature.  Nevertheless, the accuracy of the solution $\hat{\vx}$ and the fact that $\twonorm{\vx} = 1$, implies that $\twonorm{\vx} \approx 1$.

Theorems \ref{thm:fixed X} and \ref{thm:uniform} are supposed to guarantee that $\vx$ can indeed can be estimated
by a solution to \eqref{eq:optimization sparse}. But in order to apply these results, we need to know the mean width of $K_{n,s}$. 
A good bound for it follows from \eqref{eq:convexification} and Lemma \ref{lem:sparse mean width}, which give
\begin{equation}				\label{eq:K sparse mean width}
w(K_{n,s}) \leq 2 w(\conv(S_{n,s}))  \le C \sqrt{s \log(2n/s)}.
\end{equation}
This yields the following version of Corollary \ref{cor:fixed X}.

\begin{corollary}[Estimating a compressible signal]				\label{cor:sparse}
  Let $\va_1,\ldots,\va_m$ be independent standard Gaussian random vectors in $\R^n$,
  and fix $\vx \in K_{n,s}$ satisfying $\twonorm{\vx} = 1$. 
  Assume that the measurements $y_1,\ldots,y_n$ follow the model from Section \ref{sec:main results}.\footnote{Specifically, our assumptions were that 
    $\{y_i\}$ are independent random variables that are jointly independent of $\{\va_i\}$, and that \eqref{eq:glm} holds 
    with some function $\theta$ satisfying \eqref{eq:lambda}.}  
  Let $\d > 0$ and suppose that
  $$
  m \geq C \d^{-2} s \log(2n/s).
  $$
  Then, with probability at least $1 - 8 \exp(-c \delta^2 m)$,
  the solution $\hat{\vx}$ to the convex program \eqref{eq:optimization sparse} satisfies
  $$
  \twonorm{\hat{\vx} - \vx}^2 \leq \delta/\lambda.
  $$
\end{corollary}

\begin{remark}
  In a similar way, one can also specialize the uniform result, Theorem \ref{thm:uniform}, 
  to the approximately sparse case.
\end{remark}

\subsection{1-bit compressed sensing} \label{sec:1-bit CS}

Corollary \ref{cor:sparse} can be easily specialized to various specific models of noise. 
Let us consider some of the interesting models, 
and compute the correlation coefficient $\lambda = \E \theta(g) g$ in \eqref{eq:lambda} for each of them.

\begin{itemize}
  \item[] {\bf Noiseless 1-bit compressed sensing:}
    In the classic noiseless model \eqref{eq:noiseless}, the measurements are given as
    $y_i = \sign( \< \va_i, \vx \> )$ and thus $\theta(z) = \sign(z)$.  Thus
    $$
    \lambda = \E \abs{g} = \sqrt{2/\pi}.
    $$
  Therefore, with high probability we obtain $\twonorm{\hat{\vx} - \vx}^2 \leq \delta$ 
  provided that the number of measurements is $m \geq C \delta^{-2} s \log(2n/s)$.
  This is similar to the results available in \cite{pv-1-bit}.

  \item[] {\bf Random bit flips} 
  Assume that each measurement $y_i$ is only correct with probability $p$, thus
  $$
  y_i = \xi_i \sign( \< \va_i, \vx \> ), \qquad i =1 , 2, \hdots, m
  $$
  where $\xi_i$ are independent $\{-1,1\}$ valued random variables with $\Pr{\xi_i = 1}  = p$, which represent random bit flips. 
  Then $\theta(z) = \sign(z) \cdot \E \xi_1 = 2 \sign(z) (p - 1/2)$ and
  $$
  \lambda = 2 (p - 1/2) \E \abs{g} = 2 \sqrt{2/\pi} (p - 1/2).
  $$
  Therefore, with high probability we obtain $\twonorm{\hat{\vx} - \vx}^2 \leq \delta$ 
  provided that the number of measurements is 
  $m \ge C \d^{-2} (p-1/2)^{-2} s \log(2n/s)$. Thus we obtain a surprising conclusion: 
  
  \medskip
  
  \begin{quote}
    {\em The signal $\vx$ can be estimated even if each measurement is flipped with probablity nearly 1/2.} 
  \end{quote}
  
  \medskip
  
  \noindent Somewhat surprisingly, the estimation of $\vx$ is done by one simple convex program \eqref{eq:optimization sparse}.  
  Of course, if each measurement is corrupted with probability 1/2, recovery is impossible by any algorithm.
 
  \item[] {\bf Random noise before quantization}  
  Assume that the measurements are given as 
  $$
  y_i = \sign( \< \va_i, \vx \> + \nu_i), \qquad i =1 , 2, \hdots, m
  $$
  where $\nu_i$ are iid random variables representing noise added before quantization. 
  This situation is typical in analog-to-digital converters.  It is also the latent variable model from statistics.
  
  Assume for simplicity that $\nu_i$ have density $f(x)$. 
  Then $\theta(z) = 1 - 2 \Pr{\nu_i \leq -z}$, and the correlation coefficient $\lambda = \E \theta(g) g$ 
  can be evaluated using integration by parts, which gives
  $$
  \lambda = \E \theta'(g) = 2 \E f(-g) > 0.
  $$
  A specific value of $\lambda$ is therefore not hard to estimate for concrete densities $f$. For instance, 
  if $\nu_i$ are normal random variables with mean zero and variance $\s^2$, then 
  $$
  \lambda = \E \sqrt{\frac{2}{\pi \sigma^2 }} \exp(-g^2/2 \sigma^2) = \sqrt{\frac{2}{\pi (\sigma^2 + 1)}}.
  $$
  Therefore, with high probability we obtain $\twonorm{\hat{\vx} - \vx}^2 \leq \delta$
  provided that the number of measurements is $m \geq C \d^{-2} (\sigma^2 + 1)  s \log(2n/s)$.  In other words,
  \begin{equation}
  \label{eq:noisy error}
  \twonorm{\hat{\vx} - \vx}^2 \leq C \sqrt{\frac{(\sigma^2 + 1) s \log(2n/s)}{m}}.
  \end{equation}  
  Thus we obtain an unexpected conclusion: 

  \medskip

  \begin{quote}
    {\em The signal $\vx$ can be estimated even when the noise level $\sigma$ 
    eclipses the magnitude of the linear measurements.} 
  \end{quote}

  \medskip

  \noindent Indeed, the average magnitude of the linear measurements is $\E |\< \va_i, \vx \> | = \sqrt{2/\pi}$, 
  while the average noise level $\sigma$ can be much larger.  
  
  Let us also compare to the results available in the standard unquantized compressed sensing model
  \[y_i = \< \va_i, \vx \> + \nu_i \quad i = 1, 2, \hdots, m\]
 where once again we take $\nu_i \sim \NN(0, \sigma^2)$.  Under the assumption that $\onenorm{\vx} \leq \sqrt{s}$ the minimax squared error given by \cite[Theorem 1]{raskutti2009minimax}  is $\delta = c \sigma \sqrt{\frac{s \log n}{m}}$.  A slight variation on their proof yields a minimax squared error under the assumption that $m < n$ and $\vx \in K_{n,s} \cap S^{n-1}$ of $\delta = c \sigma \sqrt{ \frac{s \log(2n/s)}{m}}$.  Up to a constant, this matches the upper bound we have just derived in the 1-bit case in Equation \eqref{eq:noisy error}.  Thus we have another surprising result:
  
   \medskip

  \begin{quote}
    {\em The error in recovering the signal $\vx$ matches the minimax error for the unquantized compressed sensing problem (up to a constant):  Essentially nothing is lost by quantizing to a single bit.} 
  \end{quote}

  \medskip
\end{itemize}

\bigskip

Let us put these results in a perspective of the existing literature on 1-bit compressed sensing. 
The problem of 1-bit compressed sensing, as introduced by Boufounos and Baraniuk in \cite{Boufounos2008}, 
is the extreme version of quantized compressed sensing;
it is particularly beneficial to consider 1-bit measurements in analog-to-digital conversion  
(see the webpage \verb=http://dsp.rice.edu/1bitCS/=).  
Several numerical results are available, and there are a few recent theoretical results as well.  

Suppose that $\vx \in \R^n$ is $s$-sparse. Gupta et al. \cite{Gupta2010} demonstrate that 
the support of $\vx$ can tractably be recovered from either 
1) $O(s \log n)$ nonadaptive measurements assuming a constant dynamic range of $\vx$ 
(i.e. the magnitude of all nonzero entries of $\vx$ is assumed to lie between two constants), or 
2) $O(s \log n)$ adaptive measurements.  
Jacques et al. \cite{Jacques2011} demonstrate that any \textit{consistent} estimate of $\vx$ will be accurate provided that $m \geq O(s \log n)$.  Here consistent means that the estimate $\hat{\vx}$ should have unit norm, be at least as sparse as $\vx$, and agree with the measurements, i.e.
$\sign(\< \va_i, \hat{\vx}\> ) = \sign(\< \va_i, \vx\> )$ for all $i$.  
These results of Jacques et al. \cite{Jacques2011} can be extended to handle adversarial bit flips.  The difficulty in applying these results is that the first two conditions are nonconvex, and thus it is unknown whether there is a polynomial-time solver which is guaranteed to return a consistent solution.  We note that there are heuristic algorithms, including one in \cite{Jacques2011} which often provide such a solution in simulations.

In a dual line or research, Gunturk et al. \cite{gunturk2010, gunturk2010paper} analyze sigma-delta quantization. 
The focus of their results is to achieve an excellent dependence of on the accuracy $\d$ while minimizing the number of bits per measurement.  
However the measurements $y_i$ in sigma-delta quantization are not related to any linear measurements (unlike those in \eqref{eq:noiseless} and \eqref{eq:glm}) but are allowed to be constructed in a judicious fashion (e.g.~iteratively).
Furthermore, in Gunturk et al. \cite{gunturk2010, gunturk2010paper} the number of bits per measurement 
depends on the dynamic range of the nonzero part of $\vx$.  
Similarly, the recent work of Ardestanizadeh et al. \cite{ardestanizadeh2009} requires a finite number of bits per measurement.

The noiseless 1-bit compressed sensing given by the model \eqref{eq:noiseless} 
was considered by the present authors in the earlier paper \cite{pv-1-bit}, where the following convex program was introduced:
\begin{align*}
\min \quad &\onenorm{\vx'} \\
 \text{subject to} \quad &y_i = \sign(\< \va_i, \vx'\> ) \quad i = 1, 2, \hdots, m\\
 \text{and} \quad &\sum_{i=1}^m y_i \< \va_i, \vx' \> = m.
\end{align*}
This program was shown in \cite{pv-1-bit} to accurately recover an $s$-sparse vector $\vx$ from $m = O(s \log (n/s)^2)$ measurements $y_i$.  
This result was the first to propose a polynomial-time solver for 1-bit compressed sensing with provable accuracy guarantees. 
However, it was unclear how to modify the above convex program to account for possible noise. 

The present paper proposes to overcome this difficulty by considering the convex program \eqref{eq:optimization sparse}
(and in the most general case, the optimization problem \eqref{eq:optimization}).
One may note that the program \eqref{eq:optimization sparse} requires the knowledge of a bound on the (approximate) sparsity level $s$.  
In return, it does not need to be adjusted depending on the kind of noise or level of noise.

\subsection{Sparse logistic regression}
In order to give concrete results accessible to the statistics community, we now specialize Corollary \ref{cor:fixed X} to the {\em logistic regression} model. Further, we drop the assumption that $\twonorm{\vx} = 1$ in this section; this will allow easier comparison with the related literature (see below).  

The simple logistic function is defined as 
\begin{equation}				\label{eq:logistic function}
f (z) = \frac{e^{z}}{e^{z} + 1}.
\end{equation}
In the logistic regression model, the observations $y_i \in \{-1,1\}$ are iid random variables satisfying
\begin{equation}				\label{eq:logistic regression}
\Pr{y_i = 1} = f( \< \va_i, \vx \> ), \qquad i =1 , 2, \hdots, m.
\end{equation}
Note that this is a partial case of the generalized linear model \eqref{eq:glm} with $\theta(z) = \tanh( z/2)$. 
We thus have the following specialization of Corollary \ref{cor:fixed X}.

\begin{corollary}[Sparse logistic regression]										\label{cor:logistic regression}
  Let $\va_1,\ldots,\va_m$ be independent standard Gaussian random vectors in $\R^n$,
  and fix $\vx$ satisfying $\vx/\twonorm{\vx} \in K_{n,s}$. 
  Assume that the observations $y_1,\ldots,y_n$ follow the logistic regression model \eqref{eq:logistic regression}. 
  Let $\d > 0$ and suppose that
  $$
  m \geq C \d^{-2} s \log(2n/s).
  $$
  Then, with probability at least $1 - 8 \exp(-c \delta^2 m)$,
  the solution $\hat{\vx}$ to the convex program \eqref{eq:optimization sparse} satisfies
  \begin{equation} 			\label{eq:error with alpha}
  \twonorm{\hat{\vx} - \frac{\vx}{\twonorm{\vx}}}^2 \le \d \max(\twonorm{\vx}^{-1},1).
  \end{equation}
  \end{corollary}
 
\begin{proof}
We begin by reducing to the case when $\twonorm{\vx} = 1$ by rescaling the logistic function.  Thus, let $\alpha = \twonorm{\vx}$ and define the scaled logistic function $f_{\alpha}(\vx) = f(\alpha \vx)$.  In particular, 
$$
\Pr{y_i = 1} = f_{\alpha} \big(\< \va_i, \frac{\vx}{\twonorm{\vx}}\> \big).
$$

To apply Corollary~\ref{cor:sparse}, it suffices to compute the correlation coefficient $\lambda$ in \eqref{eq:lambda}. First, by rescaling $f$ we have also rescaled $\theta$, so we consider $\theta(z) = \tanh(\alpha z/2)$.
We can now compute $\lambda$ using integration by parts:
$$
\lambda = \E \theta(g) g = \E \theta'(g) = \frac{\alpha}{2} \E \sech^2(\alpha g/2).
$$
To further bound this quantity below, we can use the fact that $\sech^2(x)$ is an even 
and decreasing function for $x \geq 0$. This yields
\begin{align*}
\lambda &\geq \frac{\alpha}{2} \Pr{\abs{\alpha g/2} \leq 1/2} \cdot \sech^2(1/2) \\
&\geq \frac{\sech^2(1/2)}{2} \cdot \alpha \cdot \Pr{\abs{g} \leq 1/\alpha}
\geq \frac{1}{6} \min(\alpha, 1).
\end{align*}
The result follows from Corollary~\ref{cor:sparse} since $\alpha = \twonorm{\vx}$.
\end{proof}

\begin{remark}
  Corollary~\ref{cor:logistic regression} allows one to estimate the projection of $\vx$ onto the unit sphere.  
  One may ask whether the norm of $\vx$ may be estimated as well.  This depends on the assumptions made 
  (see the literature described below).  However, note that as $\twonorm{\vx}$ grows, the logistic regression model 
  quickly approaches the noiseless 1-bit compressed sensing model,
  in which knowledge of $\twonorm{\vx}$ is lost in the measurements.  
  Thus, since we do not assume that $\twonorm{\vx}$ is bounded, recovery of $\twonorm{\vx}$ becomes impossible.
\end{remark}

For concreteness, we specialized to logistic regression. 
But as mentioned in the introduction, the model \eqref{eq:glm} can be interpreted as the generalized linear model,
so our results can be readily used for various problems in sparse binomial regression. 
Some of the recent work in sparse binomial regression includes the papers 
\cite{negahban2010unified, bunea2008honest,van2008high, bach2010self, ravikumar2010high, meier2008group, kakade2009learning}.  
Let us point to the most directly comparable results.

In \cite{bach2010self, bunea2008honest, kakade2009learning, negahban2010unified} the authors propose to estimate the coefficient vector 
(which in our notation is $\vx$) by minimizing the negative log-likelihood plus an extra $\ell_1$ regularization term.  
Bunea \cite{bunea2008honest} considers the logistic regression model.  
She derives an accuracy bound for the estimate (in the $\ell_1$ norm) under a certain 
{\em condition stabil} and a under bound on the magnitude of the entries of $\vx$.  
Similarly, Bach \cite{bach2010self} and Kakade et al. \cite{kakade2009learning} derive accuracy bounds (again in the $\ell_1$ norm)
under {\em restrictive eigenvalue conditions}. 
The most directly comparable result is given by Negahban et al. \cite{negahban2010unified}.  
There the authors show that if the measurement vectors $\va_i$ have independent subgaussian entries, $\zeronorm{\vx} \leq s$, and $\twonorm{\vx} \leq 1$, 
then with high probability one has $\twonorm{\hat{\vx} - \vx}^2 \le \d$, provided that the number of measurements is 
$m \ge C \d^{-1} s \log n$.  
Their results apply to the generalized linear model \eqref{eq:glm}
under some assumptions on $\theta$.

One main novelty in this paper is that knowledge of the function $\theta$, which defines the model family, 
is completely unnecessary when recovering the coefficient vector. Indeed, the optimization problems \eqref{eq:optimization} 
and \eqref{eq:optimization sparse} do not need to know $\theta$.  This stands in contrast to programs based on maximum likelihood estimation.
This may be of interest in non-parametric statistical applications in which it is unclear which binary model to pick--the logistic model may be chosen somewhat arbitrarily.  

Another difference between our results and those above is in the conditions required.
The above papers allow for more general design matrices than those in this paper, but this necessarily leads to strong assumptions on $\twonorm{\vx}$.  As the inner products between $\< \va_i, \vx \> $ grow large, the logistic regression model approaches the 1-bit compressed sensing model. However, as shown in \cite{pv-1-bit}, accurate 1-bit compressed sensing is impossible for discrete measurement ensembles (not only is it impossible to recovery $\vx$, it is also impossible to recover $\vx/\twonorm{\vx}$).  Thus the above results, all of which \textit{do} allow for discrete measurement ensemsembles, necessitate rather strong conditions on the magnitude of $\< \va_i, \vx \> $, or equivalently, on $\twonorm{\vx}$; these are made explicitly in \cite{bunea2008honest, kakade2009learning, negahban2010unified} and implicitly in \cite{ bach2010self}.  In contrast, our theoretical bounds on the relative error only improve as the average magnitude of $\< \va_i, \vx \> $ increases.

\subsection{Low-rank matrix recovery}\label{sec:low-rank}
We quickly mention that our model applies to single bit measurements of a low-rank matrix.  Perhaps the closest practical application is quantum state tomography \cite{gross2010quantum}, but still, the requirement of Gaussian measurements is somewhat unrealistic.  Thus, the purpose of this section is to give an intuition and a benchmark.  

Let $\mX \in \R^{n_1 \times n_2}$ be a matrix of interest with rank $r$ and Froobenius norm $\fronorm{\mX} = 1$.
Consider that we have $m$ single-bit measurements following the model in the introduction so that $n = n_1 \times n_2$.  Similarly to sparse vectors, the set of low-rank matrices is not convex, but has a natural convex relaxation as follows.  Let
\[K_{n_1, n_2, r} = \{\mX \in \R^{n_1 \times n_2} : \nucnorm{\mX} \leq \sqrt{r}, \fronorm{\mX} \leq 1\}\]
where $\nucnorm{\mX}$ denotes the nuclear norm, i.e., the sum of the singular values of $\mX$.  

In order to apply Theorems \ref{thm:fixed X} and \ref{thm:uniform}, we only need to calculate $w(K_{n_1, n_2, r})$,
as follows:
\[w(K_{n_1, n_2, r}) = 2 \E \sup_{\mX \in K_{n_1, n_2, r}} \< \mG, \mX \> \]
where $\mG$ is a matrix with standard normal entries and the inner product above is the standard entrywise inner product, 
i.e. $\< \mG, \mX \> = \sum_{i,j} \mG_{i,j} \mX_{i,j}$.  Since the nuclear norm and operator norm are dual to each other, we have $\< \mG, \mX \> \leq \opnorm{\mG} \cdot \nucnorm{\mX}$.  Further, for each $\mX \in K_{n_1, n_2, r}$, $\nucnorm{\mX} \leq \sqrt{r}$, and thus
\[w(K_{n_1, n_2, r}) \leq \sqrt{r} \E \opnorm{\mG}.\]
The expected norm of a Gaussian matrix is well studied; one has $\E \opnorm{\mG} \leq \sqrt{n_1} + \sqrt{n_2}$ (see e.g., \cite[Theorem 5.32]{Vershynin2010}).  Thus, $w(K_{n_1, n_2, r}) \leq (\sqrt{n_1} + \sqrt{n_2}) \sqrt{r}$.  It follows that $O((n_1  + n_2) r)$ noiseless 1-bit measurements are sufficient to guarantee accurate recovery of rank-$r$ matrices.  We note that this matches the number of linear (infinite bit precision) measurements required in the low-rank matrix recovery literature (see \cite{candes2011tight}).

\subsection{Extension to measurements with correlated entries} \label{sec:covariance}
A commonly used statistical model would take $\va_i$ to be Gaussian vectors with correlated entries, namely $\va_i \sim \NN(0, \Sigma)$ where $\Sigma$ is a given covariance matrix.  In this section we present an extension of our results to allow such correlations.
Let $\lambda_{\min} = \lambda_{\min}(\Sigma)$ and $\lambda_{\max}=\lambda_{\max}(\Sigma)$ denote the smallest and largest eigenvalues of $\Sigma$; 
the condition number of $\Sigma$ is then $\kappa(\Sigma) = \lambda_{\max}(\Sigma) / \lambda_{\min}(\Sigma)$. It will be convenient to choose the normalization $\twonorm{\Sigma^{1/2} \vx} = 1$;
as before, this may be done by absorbing a constant into the definition of $\theta$.

We propose the following generalization of the convex program \eqref{eq:optimization sparse}: 
\begin{equation}					\label{eq:optimization with Sigma}
\begin{array}{l}
\max \sum_{i=1}^m y_i \< \va_i, \vx' \> \quad \\
\text{subject to} \quad 
\|\vx'\|_1 \leq \sqrt{s/\lambda_{\min}} \quad \text{and} \quad \|\Sigma^{1/2} \vx'\|_2 \leq 1.\end{array}
\end{equation}
The following result extends Corollary \ref{cor:sparse} to general covariance $\Sigma$. For simplicity we restrict
ourselves to exactly sparse signals; however the proof below allows for a more general signal set.

\begin{corollary}										\label{cor:sparse with Sigma}
  Let $\va_1, \ldots, \va_m$ be independent random vectors with distribution $N(0, \Sigma)$.
  Fix $\vx$ satisfying $\zeronorm{\vx} \leq s$ and $\|\Sigma^{1/2} \vx\|_2 = 1$.  
  Assume that the measurements $y_1,\ldots,y_m$ follow the model from Section~\ref{sec:main results}.  
  Let $\d>0$ and suppose that
  $$
  m \geq C \kappa(\Sigma) \, \d^{-2} s \log(2n/s).
  $$
  Then with probability at least $1 - 8 \exp(-c \delta^2 m)$, 
  the solution $\hat{\vx}$ to the convex program \eqref{eq:optimization with Sigma} satisfies
  $$
  \lambda_{\min}(\Sigma) \cdot \twonorm{\hat{\vx} - \vx}^2 
  \leq  \|\Sigma^{1/2} \hat{\vx} - \Sigma^{1/2} \vx\|_2^2
  \leq \delta / \lambda.
  $$
\end{corollary} 

\begin{remark}
  Theorem \ref{thm:uniform} can be generalized in the same way---the number of measurements required 
  is scaled by $\kappa(\Sigma)$ and the error bound is scaled by $\lambda_{\min}(\Sigma)^{-1}$.
\end{remark}

\begin{proof}[Proof of Corollary \ref{cor:sparse with Sigma}]
The feasible set in \eqref{eq:optimization with Sigma} is 
$K := \{\vx \in \R^n:\, \|\vx\|_1 \leq \sqrt{s/\lambda_{\min}(\Sigma)}, \, \|\Sigma^{1/2} \vx\|_2 \leq 1\}$.
Note that the signal $\vx$ considered in the statement of the corollary is feasible, since 
$\onenorm{\vx} \leq \twonorm{\vx} \sqrt{\zeronorm{\vx}} 
\leq \|\Sigma^{-1/2}\|\cdot \|\Sigma^{1/2} \vx\|_2 \sqrt{s} 
\leq \sqrt{s/\lambda_{\min}(\Sigma)}$.  

Define $\tilde{\va}_i := \Sigma^{-1/2} \va_i$; then $\tilde{\va_i}$ are independent standard normal vectors and 
$\< \va_i, \vx \> = \< \Sigma^{1/2} \tilde{\va}_i, \vx\> =  \< \tilde{\va}_i, \Sigma^{1/2}  \vx\> $.  
Thus, it follows from Corollary \ref{cor:fixed X} applied with $\Sigma^{1/2} \vx$ replacing $\vx$ that if 
\[m \geq C \d^{-2} w(\Sigma^{1/2} K)^2\]
then with probability at least $1 - 8 \exp(-c \d^2 m)$
$$
\|\Sigma^{1/2} \hat{\vx} - \Sigma^{1/2} \vx\|_2 \leq \d/\lambda.
$$
It remains to bound $w(\Sigma^{1/2} K)$.  Since $\Sigma^{1/2} / \|\Sigma^{1/2}\|$ acts as a contraction, Slepian's inequality 
(see \cite[Corollary~3.14]{LT})
gives $w(\Sigma^{1/2} K) \leq \|\Sigma^{1/2}\| \cdot w(K) = \lambda_{\max}(\Sigma)^{1/2} \cdot w(K)$. 
Further, $K \subseteq \lambda_{\min}(\Sigma)^{-1/2} \, K_{n,s}$.  
Thus, it follows from \eqref{eq:K sparse mean width} that 
$w(\Sigma^{1/2} K)^2 \le \kappa(\Sigma) \, w(K_{n,s}) \leq C \kappa(\Sigma) \, s \log(2n/s)$.  
This completes the proof.
\end{proof}

\section{Deducing Theorems \ref{thm:fixed X} and \ref{thm:uniform} from concentration inequalities} \label{sec:deducing}
In this section we show how to deduce our main results, Theorems \ref{thm:fixed X} and \ref{thm:uniform},
from concentration inequalities. These inequalities are stated in Propositions \ref{prop:concentration} and \ref{prop:uniform concentration} below, whose proofs are deferred to Sections~\ref{sec:concentration} and \ref{sec:uniform concentration} respectively.

\subsection{Proof of Theorem \ref{thm:fixed X}}			\label{sec:proof fixed x}
Consider the rescaled objective function from the program \eqref{eq:optimization}:
\begin{equation}
\label{eq:def f}
f_{\vx}(\vx') = \frac{1}{m} \sum_{i=1}^m y_i \< \va_i, \vx' \> .
\end{equation}
Here the subscript $x$ indicates that $f$ is a random function whose distribution depends on $x$ through $y_i$. 
Note that the solution $\hat{\vx}$ to the program \eqref{eq:optimization} satisfies $f_{\vx}(\hat{\vx}) \geq f_{\vx}(\vx)$, since $\vx$ is feasible.  
We claim that for \textit{any} $\vx' \in K$ which is far away from $\vx$, the value $f_{\vx}(\vx')$ is small with high probability.
Thus $\hat{\vx}$ must be near to $\vx$.  

To begin to substantiate this claim, let us calculate $\E f_{\vx}(\vx')$ for a \textit{fixed} vector $\vx' \in K$.

\begin{lemma}[Expectation]  \label{lem:expectation} 
  Fix $\vx \in S^{n-1}, \vx' \in B_2^n$.  Then
  $$
  \E f_{\vx}(\vx') = \lambda \< \vx, \vx'\> 
  $$
  and thus
  $$
  \E [f_{\vx}(\vx) - f_{\vx}(\vx')] = \lambda (1 - \< \vx, \vx' \> ) \geq \frac{\lambda}{2} \twonorm{\vx - \vx'}^2.
  $$
\end{lemma}

\begin{proof}
We have
$$
\E f_{\vx}(\vx') = \frac{1}{m} \sum_{i = 1}^m \E y_i \< \va_i, \vx'\> = \E y_1 \< \va_1, \vx'\> .
$$
Now we condition on $\va_1$ to give
\begin{align*}
\E y_1 \< \va_1, \vx'\> &= \E \E[y_1 \< \va_1, \vx' \> |\va_1]\\
 &= \E \theta(\< \va_1, \vx\> ) \< \va_1, \vx'\> ) .
\end{align*}
Note that $\< \va_1, \vx\> $ and $\< \va_1, \vx'\> $ are a pair of normal random variables with covariance $\< \vx, \vx' \> $.  
Thus, by taking $g, h \in \NN(0,1)$ to be independent, we may rewrite the above expectation as
\begin{align*}
&\E \theta(g) \big( \< \vx, \vx'\> g + (\twonorm{\vx'}^2 - \< \vx, \vx'\> ^2)^{1/2} \, h \big)\\
& = \< \vx, \vx'\> \E \theta(g) g =\lambda \< \vx, \vx' \> 
\end{align*}
where the last equality follows from \eqref{eq:lambda}. 
Lemma \ref{lem:expectation} is proved.
\end{proof}

Next we show that $f(\vx')$ does not deviate far from its expectation {\em uniformly} for all $\vx' \in K$.  

\begin{proposition}[Concentration]				\label{prop:concentration}
For each $t>0$, we have
\begin{align*}
&\Pr{\sup_{\vz \in K-K} \abs{f_{\vx}(\vz) - \E f_{\vx}(\vz)} \geq 4 w(K)/\sqrt{m} + t}\\
& \leq 4 \exp(-m t^2/8).
\end{align*}
\end{proposition}
This result is proved in Section \ref{sec:concentration} using standard techniques of probability in Banach spaces. 

Theorem \ref{thm:fixed X} is a direct consequence of this proposition.

\begin{proof}[Proof of Theorem \ref{thm:fixed X}]
Let $t >0$. 
By Proposition \ref{prop:concentration}, the following event occurs with probability at least $1 - 4 \exp(-m t^2/8)$:
\[
\sup_{\vz \in K} \abs{f_{\vx}(\vz) - \E f_{\vx}(\vz)} \leq 4 w(K) /\sqrt{m} + t.
\]
Suppose the above event indeed occurs. Let us apply this inequality for $\vz = \hat{\vx} - \vx \in K-K$.
By definition of $\hat{\vx}$, we have $f_{\vx}(\hat{\vx}) \geq f_{\vx}(\vx)$. 
Noting that the function $f_{\vx}(\vz)$ is linear in $\vz$, we obtain 
\begin{align}			
&0 \le f_{\vx}(\hat{\vx}) - f(\vx) 
= f_{\vx}(\hat{\vx} - \vx) \nonumber \\
&\le \E [f_{\vx}(\hat{\vx} - \vx)] + 4w(K)/\sqrt{m} + t \nonumber \\
&\le -\frac{\lambda}{2} \twonorm{\hat{\vx} - \vx}^2 + 4w(K)/\sqrt{m} + t.\label{eq:finish fixed}
\end{align}
The last inequality follows from Lemma \ref{lem:expectation}.  
Finally, we choose $t = 4\beta/\sqrt{m}$ and rearrange terms to complete the proof of Theorem \ref{thm:fixed X}.
\end{proof}

\subsection{Proof of Theorem \ref{thm:uniform}}
The argument is similar to that of Theorem~\ref{thm:fixed X} given above. 
We consider the rescaled objective functions with corrupted and uncorrupted measurements: 
\begin{equation}							\label{eq:corrupted uncorrupted}
\begin{array}{rl}
f_{\vx}(\vx') & = \frac{1}{m} \sum_{i=1}^m y_i \< \va_i, \vx' \> , \\ 
\tilde{f}_{\vx}(\vx') & = \frac{1}{m} \sum_{i=1}^m \tilde{y}_i \< \va_i, \vx' \> \\
  &= \frac{1}{m} \sum_{i=1}^m \sign(\< \va_i, \vx\> ) \< \va_i, \vx' \>
  \end{array}
\end{equation}
Arguing as in Lemma \ref{lem:expectation} (with $\theta(z) = \sign(z)$), we have
\begin{equation}							\label{eq:expected f}
\E \tilde{f}_{\vx}(\vx') = \lambda \< \vx, \vx'\> = \E \abs{g} \cdot \< \vx, \vx'\> = \sqrt{2/\pi} \, \< \vx, \vx'\> .
\end{equation}
Similarly to the proof of Theorem \ref{thm:fixed X}, 
we now need to show that $f_{\vx}(\vx')$ does not deviate far from the expectation of $\tilde{f}_{\vx}(\vx')$; 
but this time the result should hold uniformly over not only $\vx' \in K - K$ but also $\vx \in K$ 
and $\vy$ with small Hamming distance to $\tilde{\vy}$.  This is the content of the following proposition.

\begin{proposition}[Uniform Concentration] 		\label{prop:uniform concentration} 
  Let $\d>0$ and suppose that
  $$
  m \geq C \delta^{-6} w(K)^2.
  $$
  Then with probability at least $1 - 8 \exp (- c \delta^2 m)$, we have
  \begin{equation}				\label{eq:prop uniform bound}
  \sup_{\vx, \vz, \vy} \abs{f_{\vx}(\vz) - \E \tilde{f}_{\vx}(\vz)} \leq \delta \sqrt{\log(e/\d)} +  4 \tau \sqrt{\log(e/\tau)}
  \end{equation}
  where the supremum is taken over $\vx \in K \cap S^{n-1}$, $\vz \in K-K$ and $\vy$ satisfying $d_H(\vy, \tilde{\vy}) \leq \tau m$.
\end{proposition}

This result is significantly deeper than Proposition~\ref{prop:concentration}. 
It is based on a recent geometric result from \cite{pv-embeddings} on random tessellations of sets on the sphere.
The proof is given is proved in Section \ref{sec:uniform concentration}. 

Theorem~\ref{thm:uniform} now follows from the same steps as in the proof of Theorem \ref{thm:fixed X} above. \qed

\begin{remark}[Random noise]				\label{rem:random noise}
  The adversarial bit flips allowed in Theorem~\ref{thm:uniform} can be combined with 
  random noise. We considered two models of random noise in Section~\ref{sec:1-bit CS}.
  One was {\em random bit flips} where one would take $\tilde{y}_i = \xi_i \sign(\< \va_i, \vx\> )$;
  here $\xi_i$ are iid Bernoulli random variables satisfying $\Pr{\xi_i = 1} = p$.  
  The proof of Theorem \ref{thm:uniform} would remain unchanged under this model, aside from the calculation
  $$
  \E f_{\vx}(\vx') = \sqrt{2/\pi} \, (2p -1).
  $$
  The end result is that the error bound in \eqref{eq:uniform error bound} would be divided by $2p - 1$.

  Another model considered in Section~\ref{sec:1-bit CS} was {\em random noise before quantization}. 
  Thus we let 
  \begin{equation}\label{eq:another model}
  \tilde{y}_i = \sign( \< \va_i, \vx \> + g_i)
  \end{equation} 
  where $g_i \sim \NN(0, \sigma^2)$ are iid. 
  Once again, a slight modification of the proof of Theorem \ref{thm:uniform} allows the incorporation of such noise. 
  Note that the above model is equivalent to $\tilde{y}_i = \sign( \< \tilde{\va}_i , \tilde{\vx} \> )$
  where $\tilde{\va}_i = (\va_i, g_i)$ and $\tilde{\vx} = (\vx, \sigma)$ (we have concatenated an extra entry onto $\va_i$ and $\vx$).  Thus, by slightly adjusting the set $K$, we are back in our original model.
\end{remark}

\section{Concentration: proof of Proposition~\ref{prop:concentration}}				\label{sec:concentration}

Here we prove the concentration inequality given by Proposition~\ref{prop:concentration}.

\subsection{Tools: symmetrization and Gaussian concentration}
The argument is based on the standard techniques of probability in Banach spaces---symmetrization and the Gaussian concentration inequality.
Let us recall both these tools.

\begin{lemma}[Symmetrization]							\label{lem:symmetrization}
  Let $\e_1, \e_2, \hdots, \e_m$ be independent Rademacher random variables.\footnote{This means that 
    $\P\{\e_i = 1\} = \P\{\e_i = -1\} = 1/2$ for each $i$. The random variables $\e_i$ are assumed to be independent of each other and of any other 
    random variables in question, namely $\va_i$ and $y_i$.}
  Then
  \begin{align}  
  \mu := \E \sup_{\vz \in K-K} |f_{\vx}(\vz) - \E f_{\vx}(\vz)| 
  \leq 2 \E \sup_{\vz \in K-K} \frac{1}{m} \abs{\sum_{i=1}^m \e_i y_i \< \va_i, \vz\> }.\label{eq:expected symmetrization}
  \end{align}
  Furthermore, we have the deviation inequality
  \begin{align}					
  \Pr{\sup_{\vz \in K-K} \abs{f(\vz) - \E f(\vz)} \geq 2 \mu + t} 
  \leq 4\Pr{\sup_{\vz \in K-K} \abs{\sum_{i=1}^m \e_i y_i \< \va_i, \vz \>} > t/2}. \label{eq:probability-symmetrization}
  \end{align}
\end{lemma}
Inequality \eqref{eq:expected symmetrization} follows e.g., from the proof of \cite[Lemma 6.3]{LT}.  
The proof of inequality \eqref{eq:probability-symmetrization} is contained in \cite[Chapter 6.1]{LT}. \qed

\medskip

\begin{theorem}[Gaussian concentration inequality]  \label{thm:Gaussian concentration}
  Let $(G_{\vx})_{\vx \in T}$ be a centered Gaussian process indexed by a finite set $T$.
  Then for every $r>0$ one has
  $$
  \Pr{\sup_{\vx \in T} G_{\vx}  \geq \E \sup_{\vx \in T} G_{\vx} + r} \leq \exp(-r^2/2\sigma^2)
  $$
  where $\sigma^2 = \sup_{\vx \in T} \E G_{\vx}^2 < \infty$.
\end{theorem}
\noindent A proof of this result can be found e.g.~in \cite[Theorem~7.1]{Ledoux}. \qed

This theorem can be extended to separable sets $T$ 
in metric spaces by an approximation argument. In particular, given a set $K \subseteq B_2^n$ and $r > 0$,
the standard Gaussian random vector $\vg$ in $\R^n$ satisfies
\begin{equation}							\label{eq:Gaussian concentration}
\Pr{\sup_{\vx \in K-K} \< \vg, \vx \> \geq w(K) + r} \leq \exp(-r^2/2).
\end{equation}

\medskip

\subsection{Proof of Proposition \ref{prop:concentration}}
We apply the first part \eqref{eq:expected symmetrization} of Symmetrization Lemma~\ref{lem:symmetrization}. Note that 
since $\va_i$ have symmetric distributions and $y_i \in \{-1,1\}$, the random vectors 
$\e_i y_i \va_i$ has the same (iid) distribution as $\va_i$.  Using the rotational invariance and the symmetry of the Gaussian distribution,
we can represent the right hand side of \eqref{eq:expected symmetrization} as 
\begin{align}				
&\hspace{-2mm}\sup_{\vz \in K-K} \frac{1}{m} \abs{\sum_{i=1}^m \e_i y_i \< \va_i, \vz\> } 
\stackrel{dist}{=}\sup_{\vz \in K-K} \frac{1}{m} \abs{\sum_{i=1}^m \< \va_i, \vz\> } \nonumber\\ 
&\stackrel{dist}{=} \frac{1}{\sqrt{m}} \sup_{\vz \in K-K} \abs{\< \vg, \vz\> }
= \frac{1}{\sqrt{m}} \sup_{\vz \in K-K} \< \vg, \vz\> \label{eq:decouple}
\end{align}
where $\stackrel{dist}{=}$ signifies the equality in distribution. 
So taking the expectation in \eqref{eq:expected symmetrization} we obtain 
\begin{align}							
\E\sup_{\vz \in K-K} |f_{\vx}(\vz) - \E f_{\vx}(\vz)|  
\le \frac{2}{\sqrt{m}} \, \E\sup_{\vz \in K-K} \< \vg, \vz\>
= \frac{2w(K)}{\sqrt{m}}.\label{eq:expected deviation}
\end{align}

To supplement this expectation bound with a deviation inequality, 
we use the second part \eqref{eq:probability-symmetrization} of Symmetrization Lemma~\ref{lem:sparse mean width} 
along with \eqref{eq:expected deviation} and \eqref{eq:decouple}. This yields
$$
\Pr{\sup_{\vz \in K-K} |f_{\vx}(\vz) - \E f_{\vx}(\vz)| \geq \frac{4w(K)}{\sqrt{m}} + t} 
\leq 4 \Pr{\frac{1}{\sqrt{m}} \sup_{\vz \in K-K} \< \vg, \vz\> > t/2}.
$$
Now it remains to use the Gaussian concentration inequality \eqref{eq:Gaussian concentration} with 
$r = t \sqrt{m}/2$. The proof of Proposition \ref{prop:concentration} is complete.
\qed

\section{Concentration: proof of Proposition~\ref{prop:uniform concentration}}					\label{sec:uniform concentration}

Here we prove the uniform concentration inequality given by Proposition~\ref{prop:uniform concentration}.
Beside standard tools in geometric functional analysis such as Sudakov minoration for covering numbers, 
our argument is based on the recent work \cite{pv-embeddings} on random hyperplane tessellations. 
Let us first recall the relevant tools. 

\subsection{Tools: covering numbers, almost isometries and random tessellations}

Consider a set $T \subset \R^n$ and a number $\e>0$. Recall that an {\em $\e$-net of $T$} (in the Euclidean norm)
is a set $N_\e \subset T$ which has the following property:
for every $\vx \in T$ there exists $\bar{\vx} \in N_\e$ satisfying $\twonorm{\vx - \bar{\vx}} \leq \e$.
The {\em covering number of $T$} to precision $\e$, which we call $N(T, \e)$, is the minimal cardinality of an $\e$-net of $T$.
The covering numbers are closely related to the mean width, as shown by the following well-known inequality:

\begin{theorem}[Sudakov Minoration] \label{thm:sudakov} 
  Given a set $T \subset \R^n$ and a number $\e > 0$, one has
  $$
  \log N(K, \e) \le \log N(K-K, \e) \leq C \e^{-2} w(K)^2.
  $$
\end{theorem}
\noindent A proof of this theorem can be found e.g.~in \cite[Theorem 3.18]{LT}.		\qed

\medskip

We will also need two results from the recent work \cite{pv-embeddings}. 
To state them conveniently, let $\mA$ denote the $m \times n$ matrix with rows $\va_i$. 
Thus $\mA$ is standard Gaussian matrix with iid standard normal entries. 
The first (simple) result guarantees that $\mA$ acts as a almost isometric embedding from $(K, \twonorm{\cdot})$ into $(\R^m, \onenorm{\cdot})$.  

\begin{lemma}[\cite{pv-embeddings} Lemma~2.1]			\label{lem:l1 embedding}  
  Consider a subset $K \subset B_2^n$. Then for every $u > 0$ one has
  $$
  \Pr{\sup_{\vx \in K-K} \abs{\frac{1}{m} \onenorm{\mA \vx} - \sqrt{\frac{2}{\pi}} \twonorm{\vx}} \geq \frac{4 w(K)}{\sqrt{m}} + u}
  \leq 2 \exp\left( - \frac{m u^2}{2}\right).
  $$
\end{lemma}

The second (not simple) result demonstrates that the discrete map $\vx \mapsto \sign(\mA \vx)$ acts as an almost isometric
embedding from ($K \cap S^{n-1}, d_G$) into the Hamming cube $(\{-1,1\}^m, d_H)$. Here $d_G$ and $d_H$ denote the geodesic and Hamming 
metrics respectively; the sign function is applied to all entries of $\mA \vx$, thus $\sign(\mA \vx)$ is the vector
with entries $\sign(\< \va_i, \vx\> )$, $i=1,\ldots,m$.

\begin{theorem}[\cite{pv-embeddings} Theorem 1.5]			\label{thm:tessellations}
  Consider a subset $K \subset B_2^n$ and let $\d > 0$. Let 
\[m \geq C \delta^{-6} w(K)^2.\]
Then with probability at least $1 - 2 \exp(-c \delta^2 m)$, the following holds for all $\vx, \vx' \in K\cap S^{n-1}$ :
\[\abs{\frac{1}{\pi} d_G(\vx, \vx') - \frac{1}{m} d_H(\sign(\mA \vx), \sign(\mA \vx'))} \leq \delta.\]
\end{theorem}

\medskip

\subsection{Proof of  Proposition \ref{prop:uniform concentration}}
Let us first assume that $\tau = 0$ for simplicity; the general case will be discussed at the end of the proof. 
With this assumption, \eqref{eq:corrupted uncorrupted} becomes
\begin{equation}							\label{eq:f tildef}
f_{\vx}(\vz) = \tilde{f}_{\vx}(\vz) 
= \frac{1}{m} \sum_{i=1}^m \sign(\< \va_i, \vx\> ) \< \va_i, \vz \> .
\end{equation}
To be able to approximate $\vx$ by a net, we use Sudakov minoration (Theorem \ref{thm:sudakov})
and find a $\d$-net $N_{\d}$ of $K \cap S^{n-1}$ whose cardinality satisfies
\begin{equation}							\label{eq:N delta}
\log\abs{N_\d} \leq C \d^{-2} w(K)^2.
\end{equation}

\begin{lemma}				\label{lem:helpful events}
  Let $\d>0$ and assume that $m \geq C \d^{-6} w(K)^2$.  Then we have the following.
  \begin{enumerate}[1.] 	
    \item (Bound on the net.) 		\label{part:on net}
    With probability at least $1 - 4\exp(-c m \d^2)$ we have
    \begin{equation}			\label{eq:concentration on cover}
    \sup_{\vx_0, \vz} \abs{f_{\vx_0}(\vz) - \E f_{\vx_0}(\vz)} \leq \delta
    \end{equation}
    where the supremum is taken over all $\vx_0 \in N_\d$ and all $\vz \in K-K$.

    \item (Deviation of sign patterns.)  		\label{part:sign patterns}
    For $\vx, \vx_0 \in K \cap S^{n-1}$, consider the set   
    $$
    T(\vx, \vx_0) := \{i \in [m]: \sign(\< \va_i, \vx \> ) \neq \sign(\< \va_i, \vx_0 \> ) \}.
    $$
    Then, with probability at least $1 - 2 \exp(-c m \d^2)$ we have
    \begin{equation}			\label{eq:deviation of signs}
    \sup_{\vx, \vx_0} |T(\vx, \vx_0)| \leq 2m \d
    \end{equation}
    where the supremum is taken over all $\vx, \vx_0 \in K \cap S^{n-1}$ satisfying $\twonorm{\vx - \vx_0} \leq \d$.

    \item (Deviation of sums.)  		\label{part:deviation of sums}
    Let $s$ be a natural number.  With probability at least $1 - 2 \exp(- s \log(em/s)/2)$ we have
    \begin{align}		
    \sup_{\vz, T} \sum_{i \in T} |\< \va_i, \vz \> |
    \leq s \left(\sqrt{\frac{8}{\pi}} + \frac{4 w(K)}{\sqrt{s}} + 2 \sqrt{\log(em/s)}\right).\label{eq:deviation of summation}
    \end{align}
    where the supremum is taken over all $\vz \in K-K$ and all subsets $T \subset [n]$ with cardinality $\abs{T} \leq s$.
\end{enumerate}
\end{lemma}

\begin{proof}[Proof of Proposition \ref{prop:uniform concentration}]
Let us apply Lemma~\ref{lem:helpful events}, where in part~\ref{part:deviation of sums} 
we take $s=2 m \d$ rounded down to the next smallest integer.  
Then all of the events \eqref{eq:concentration on cover}, \eqref{eq:deviation of signs}, \eqref{eq:deviation of summation}
hold simultaneously with probability at least $1 - 8 \exp(-c m \d^2)$.

Recall that our goal is to bound the deviation of $f_{\vx}(\vz)$ from its expectation uniformly over $\vx \in K \cap S^{n-1}$ and $\vz \in K-K$.
To this end, given $\vx \in K \cap S^{n-1}$ we choose $\vx_0 \in N_\d$ so that $\twonorm{\vx - \vx_0} \leq \d$.  
By \eqref{eq:f tildef} and definition of the set $T(\vx, \vx_0)$ in Lemma~\ref{lem:helpful events}, we can approximate
$f_{\vx}(\vz)$ by $f_{\vx_0}(\vz)$ as follows:
\begin{equation}			\label{eq:bound by T}
\abs{f_{\vx}(\vz) - f_{\vx_0}(\vz)} 
\leq \frac{2}{m} \sum_{i \in T(\vx, \vx_0)} \abs{\< \va_i, \vz\> }.
\end{equation}
Furthermore, \eqref{eq:deviation of signs} guarantees that $|T(\vx, \vx_0)| \leq 2m \delta$.  
It then follows from \eqref{eq:deviation of summation}, our choice $s=2 m \d$ and the assumption on $m$ that
$$
\sum_{i \in T(\vx, \vx_0)} \abs{\< \va_i, \vz\>} \leq C m \delta \sqrt{\log(e/\delta)}.
$$
Thus
$$
\abs{f_{\vx}(\vz) - f_{\vx_0}(\vz)} \leq 2C \delta \sqrt{\log(e/\delta)}.
$$
Combining this with \eqref{eq:concentration on cover} we obtain
\begin{equation}							\label{eq:fx Efx0}
\abs{f_{\vx}(\vz) - \E f_{\vx_0}(\vz)} \leq C_1 \delta \sqrt{\log(e/\delta)}.
\end{equation}
Further, recall from \eqref{eq:expected f} that $\E f_{\vx_0}(\vz) = \sqrt{2/\pi} \, \< \vx, \vz \> $ and thus
\begin{align}							
&|\E f_{\vx_0}(\vz) - \E f_{\vx}(\vz)| 
= \sqrt{2/\pi} \, \abs{\< \vx_0 - \vx, \vz \> } \nonumber \\
&\leq 2 \sqrt{2/\pi} \twonorm{\vx_0 - \vx} 
\leq 2 \sqrt{2/\pi} \, \d.\label{eq:fx0 fx}
\end{align}
The last two inequalities in this line follow since $\vz \in K-K \subset 2B_2^n$ and $\|\vx_0 - \vx\|_2 \le \d$.
Finally, we combine inequalities \eqref{eq:fx Efx0} and \eqref{eq:fx0 fx} to give
$$
\abs{f_{\vx}(\vz) - \E f_{\vx}(\vz)} \leq C_2 \delta \sqrt{\log(e/\delta)}.
$$
Note that we can absorb the constant $C_2$ into the requirement $m \geq C \d^{-6} w(K)^2$. 
This completes the proof in the case where $\tau=0$. 

In the general case, we only need to tweak the above argument by increasing the size of $s$ 
considered in \eqref{eq:deviation of summation}. Specifically, it is enough to choose $s$ 
to be $\tau m + 2 m \d$ rounded down to the next smallest integer.
This allows one to account for arbitrary $\tau m$ bit flips of the numbers $\sign(\< \va_i, \vx\> )$, 
which produce the difference between $f_{\vx}(\vz)$ and $\tilde{f}_{\vx}(\vz)$.
The proof of Proposition \ref{prop:uniform concentration} is complete.
\end{proof}

It remains to prove Lemma~\ref{lem:helpful events} that was used in the argument above. 

\begin{proof}[Proof of Lemma \ref{lem:helpful events}]  
To prove part~\ref{part:on net}, we can use Proposition~\ref{prop:concentration} combined with the union bound
over the net $N_\delta$. Using the bound \eqref{eq:N delta} on the cardinality of $N_\delta$ we obtain 
\begin{align*}
&\Pr{\sup_{\vx_0, \vz} \abs{f_{\vx_0}(\vz) - \E f_{\vx_0}(\vz)} \geq 4 w(K)/\sqrt{m} + t} \\
  &\leq \abs{N_\d} 4 \exp(-m t^2/8)\\
&\leq 4 \exp \big(-m t^2/8 + C \d^{-2} w(K)^2 \big)
\end{align*}
where the supremum is taken over all $\vx_0 \in N_\d$ and $\vz \in K-K$.  
It remains to chose $t = \d/2$ and recall that $m \geq C \delta^{-6} w(K)^2$ to finish the proof.

We now turn to part~\ref{part:sign patterns}.  
First, note that $\abs{T(\vx,\vx_0)} = d_H (\sign(\mA \vx), \sign(\mA \vx_0))$.  
Theorem \ref{thm:tessellations} demonstrates that this Hamming distance is almost isometric to the geodesic distance,
 which itself satisfies $\frac{1}{\pi} d_G(\vx, \vx_0) \leq \twonorm{\vx - \vx_0} \leq \delta$.  
Specifically, Theorem \ref{thm:tessellations} yields that under our assumption that $m \geq C \delta^{-6} w(K)^2$, 
with probability at least $1 - 2 \exp(-c \delta^2 m)$ one has 
\begin{equation}		\label{eq:bound T}
\abs{T(\vx,\vx_0)} \leq 2m\delta
\end{equation}
for all $\vx,\vx_0 \in K \cap S^{n-1}$ satisfying $\twonorm{\vx - \vx_0} \leq \delta$. This proves part~\ref{part:sign patterns}.

In order to prove part  \ref{part:deviation of sums}, we may consider the subsets $T$ satisfying $|T| = s$; 
there are ${m \choose s} \leq \exp(s \log(e m/s))$ of them.
Now we apply Lemma~\ref{lem:l1 embedding} for the $T \times n$ matrix $\mP_T \mA$ where 
$\mP_T$ denotes the coordinate restriction in $\R^m$ onto $\R^T$; so in the statement of Lemma~\ref{lem:l1 embedding}
we replace $m$ by $|T|=s$. Combined with the union bound over all $T$, this gives
$$
\Pr{\sup_{\vz, T} \frac{1}{s} \sum_{i\in T} \abs{\< \va_i, \vz\>} 
\geq \sqrt{\frac{2}{\pi}} \twonorm{\vz} + \frac{4 w(K)}{\sqrt{s}} + u} 
\leq 2 \exp\left(s \log(e m /s) - \frac{s u^2}{2}\right).
$$
Recall that $\|\vz\|_2 \le 2$ since $\vz \in K-K \subset 2B_2^n$. Finally, 
we take $u^2 = 4 \log(em/s)$ to complete the proof.
\end{proof}

\section{Discussion} \label{sec:conclusion}
Unlike traditional compressed sensing, which has already enjoyed an extraordinary wave of theoretical results, 
1-bit compressed sensing is in its early stages.  In this paper, we proposed a polynomial-time solver (given by a convex program)
for noisy 1-bit compressed sensing, and we gave theoretical guarantees on its performance. 
The discontinuity inherent in 1-bit measurements led to some unique mathematical challenges.  
We also demonstrated the connection to sparse binomial regression, and derived novel results for this problem as well.

The problem setup in 1-bit compressed sensing (as first defined in \cite{Boufounos2008}) is quite elegant, allowing for a theoretical approach.  
On the other hand, there are many compressed sensing results assuming substantially finer quantization.  
It would be of interest to build a bridge between the two regimes; for example, 2-bit compressed sensing would already open up new questions. 

\bibliographystyle{acm}
\bibliography{pv-1bitcs-robust}
%

\end{document}